\documentclass{article}
\usepackage{graphicx} 

\usepackage{amssymb}
\usepackage{authblk}
\usepackage{amsmath}
\usepackage{amsthm}
\usepackage{comment}
\usepackage{prodint}
\usepackage{appendix}
\usepackage{subcaption}
\usepackage{tablefootnote}
\usepackage{longtable}
\usepackage{float}
\usepackage[style=apa,sorting=nyt,maxbibnames=99]{biblatex}
\makeatletter
\renewcommand\@biblabel[1]{}
\makeatother

\newcommand{\conp}{\overset{P}{\rightarrow}}

\newtheorem{theorem}{Theorem}

\title{A non-parametric estimator for excess recurrent events}

\author{Jonatan Hedberg$^{1}$, Ola Hössjer$^{2}$, Caroline Nordenvall$^{3}$ Therese M-L Andersson$^{1}$, Elisavet Syriopoulou$^{1}$   \\
        \small $^{1}$Department of Medical Epidemiology and Biostatistics, Karolinska Institutet, Stockholm, Sweden \\
        \small $^{2}$Department of Mathematics, Stockholm University, Stockholm, Sweden \\
        \small $^{3}$Department of Molecular Medicine and Surgery, Karolinska Institutet, Stockholm, Sweden \\
}

\begin{document}
\maketitle

\begin{abstract}
\noindent
Measuring disease burden is an important part of both public health research and health economics. Some measures of disease burden, such as hospitalisations, are made up of recurrent events. Assessing what events are related to a particular disease is however non-trivial. We extend the notion of relative survival to recurrent events by developing a novel non-parametric estimator for the expected cumulative excess number of recurrent events experienced by patients with a specific disease. The estimator utilises data from some cohort with aggregated population level data to estimate the number of excess recurrent events. Using empirical process theory, we show that the estimator converges weakly to a mean zero Gaussian process under mild regularity conditions, and provide a consistent estimator for the covariance function. We also evaluate the finite sample properties of the estimator through simulations and provide a practical example using data from Swedish patients with rectal cancer. 
\newline 

\end{abstract}
\section{Introduction}
Hospitalisations and other recurrent events occur frequently in epidemiological data. The events are often expensive, generating a burden on the medical system and debilitating to the patients suffering them. As such there is a great societal need in being able to quantify and understand what diseases cause many events and thereby need greater intervention or prevention.

Reasons for recurring clinical events such as hospitalisations can often be missing or unreliably recorded in registry data. The number of hospitalisations related to a particular condition might therefore be difficult to estimate. Even in situations where the reason for hospitalisation is known, there is no guarantee that the recorded reason accurately captures the disease burden. It might be that the condition or its treatment increases the risk of other types of hospitalisations where the primary reason is something else. For example, a patient with cancer undergoing chemotherapy, may have several hospitalisations due to serious infections that may be a result of the treatment. Only measuring the hospitalisations that are labelled as cancer-related would then result in an underestimation of the disease burden. Therefore, there is a need for methods and concepts that can bypass this problem.

One approach is to use general methods for recurrent events such as the 
estimator proposed by Ghosh \& Lin \parencite*{gosh2000} or some regression model \parencite{Lin2000, Ghosh2003, Mao2015} and compare directly with controls sampled from the general population. Sometimes however, ethical or legal reasons, or simply the difficulty in defining suitable comparators might prohibit access to controls.

Relative survival methods deal with the issue in the survival analysis setting. For these methods population level mortality data (usually from life tables) are used to compare the observed survival in the cohort currently being analysed to the expected survival in the absence of the disease and obtain a measure of excess mortality following a disease diagnosis. There is much literature on this subject. Several non-parametric estimators for relative survival exist \parencite{pohar,andersen1989}, and the concept has been extended to regression models both semi-parametric \parencite{sasieni1996} and parametric \parencite{nelson2007}.

To our knowledge, the literature on methods for how to handle relative recurrent events is rather limited. Manevski et al. \parencite*{Manevski2022} and Weibull et al. \parencite{Weibull2021} consider multi-state model that incorporate relative survival estimators. However neither of the two papers describe the situation where the same type of event can happen multiple times, and only the second uses a relative survival estimator for some non-fatal event. For this reason, we propose a non-parametric estimator that allows for the estimation of the excess number of recurrent events in relation to what would be expected in a population without the disease in question. This allows for the study and quantification of disease burden without the need for a control group for the comparison. Instead tables of event rates from the general population are required. 

This article is structured as follows: first in section \ref{sec:est} we present the estimator with necessary assumptions. We also present three theorems detailing the weak convergence of the estimator to a zero mean Gaussian process and how to estimate its asymptotic covariance function. The proofs of the theorems can be found in the appendix. Section \ref{sec:sim} presents a simulation study where we evaluate the finite sample properties of the estimator. The method is then illustrated using real data in section \ref{sec:app}. Lastly a short discussion is provided in section \ref{sec:disc}.

\section{The non-parametric estimator} \label{sec:est}
Assume a sample of size $n$ and let $N^{R*}_i(t)$ be the number of recurrent events of individual $i$ up until and including time $t$ for $i=1,...,n$. Let $D_i$ be the time of death of individual $i$, and let $N^{D*}_i(t)=1(D_i\le t)$ denote whether (1) or not (0) $i$ has died at time $t$. Define a fixed end of follow-up $\tau$ so that we only consider $t\in[0,\tau]$. 
We assume right censoring, $C_i\wedge\tau$, where $C_i$ is a random variable and that $C_i\perp N_i^{R*},N_i^{D*}$. Due to censoring we may not observe the full number of events and therefore define the stopped processes $N_i^R(t):=N_i^{R*}(t\wedge C_i)$ and $N_i^D(t):=N^{D*}_i(t\wedge C_i)$. Note that the assumption of independent censoring is not completely necessary. It can be replaced with an assumption regarding censoring being non-informative on the excess number of events.

Assume that we have independent and identically distributed (iid) characteristics $X_i$ and $Z_i$ where $X_i$ represent baseline characteristics of individual $i$ known to the investigator, and $Z_i$ disease specific characteristics. It is reasonable to assume that both $X_i$ and $Z_i$ affect the rate of recurrent events as well as death. Assuming an additive excess rate model this leads us to the following recurrent events rate and decomposition of individual $i$: $d\Lambda^R_i(t)=d\Lambda^R(t;X_i,Z_i)=d\Lambda_{E}^R(t;X_i,Z_i)+d\Lambda_{P}^R(t;X_i)$, where subscript $E$ denotes the excess rate and $P$ the general population rate. As such, the recurrent events rates are functions of random variables.
In the case of $d\Lambda_P(t;X)$ we shall assume that it is a known function of $t$ and $X$ (e.g.\ from estimates from population level data). For mortality we could assume a similar decomposition, but since we are interested in the excess recurrent events rate as long as the individual is alive, the reason for their death is of lesser concern. Thus, we will be satisfied by formulating the following model for the hazard rate: $d\Lambda^D_i(t)=d\Lambda^D(t;X_i,Z_i)$, with accompanying survival curve $S(t;X_i,Z_i)=\exp[-\int_0^t d\Lambda^D(u;X_i,Z_i)]$. 

Let $Y_i(t)=1[\min(C_i,D_i)\ge t]$ be an indicator of whether $i$ is alive and not censored at time $t$. Since $D_i$ and $C_i$ are independent, it follows that $E[Y_i(t)]=C(t)S(t)$, where $C(t)=P(C_i\ge t)$ and $S(t)=P(D_i\ge t)=E[S(t;X_i,Z_i)]$. The excess recurrent events rates and population level rates of a randomly chosen individual from the population, at time $t$, are given by  
$d\Lambda_E^R(t) = E[Y_i(t)d\Lambda_E^R(t;X_i,Z_i)]/[C(t)S(t)]$ 
and $d\Lambda_P^R(t) = E[Y_i(t)d\Lambda_P^R(t;X_i)]\\ /[C(t)S(t)]$ 
respectively. The hazard rate at time $t$ of a randomly chosen individual from population is given by $d\Lambda^D(t)=E[Y_i(t)d\Lambda^D(t;X_i,Z_i)]/[C(t)S(t)]$.

Following the same thinking as in Ghosh \& Lin \parencite*{gosh2000}, we are interested in knowing how many recurrent events we expect to occur while individuals are still alive. For this reason we want to estimate the expected number 
\begin{align}
    \label{eq:estimand}
    \mu_E(t)&=\int_0^td\Lambda_{E}^R(u)S(u)
\end{align}
of excess recurrent events up to time $t$. This is the marginal number of excess recurrent events that a randomly sampled individual with the disease suffers during $[0,t]$ as long as they are alive (whereafter they can have no such events). 

We propose estimating $\mu_E(t)$ by plugging in marginal non-parametric estimates of $S(t)$ and $d\Lambda_{E}^R(t)$ into (\ref{eq:estimand}). These estimates are given by the product limit estimator
\begin{align} \label{eq:kmcomp}
    \hat S(t)=\prodi_{[0,t]}\left(1-\frac{\sum_{i=1}^n  dN^D_i(s)}{\sum_{i=1}^nY_i(s)}\right)=\prodi_{[0,t]}\left(1-d\hat\Lambda^D(s)\right)
\end{align}
and
\begin{align} \label{eq:nacomp}
    d\hat\Lambda^R_{E}(t)=\sum_{i=1}^n \frac{dN^R_i(t)-Y_i(t)d\Lambda^R_{P}(t;X_i)}{\sum_{j=1}^nY_j(t)},
\end{align}
respectively. The estimator (\ref{eq:nacomp}) mirrors the estimator of Ghosh \& Lin \parencite*{gosh2000} apart from our recurrent events component relating to the excess recurrent events. Incorporating (\ref{eq:kmcomp}) and (\ref{eq:nacomp}) into (\ref{eq:estimand}) we obtain an estimate 
\begin{align}
\label{eq:estimate}
\hat{\mu}_E(t) = \int_0^t d\hat\Lambda_{E}^R(u)\hat S(u)    
\end{align}
of $\mu_E(t)$. It follows from theorem \ref{thrm:est} below and the Continuous Mapping Theorem that

\begin{align}
     \label{eq:unif}
     \sup_{t\in [0,\tau]}\left|\hat{\mu}_E(t)-\mu_E(t) \right| \conp0.
\end{align}
In other words, even though we simply plug in marginal estimates of $S(t)$ and $\Lambda_E^R(t)$ into (\ref{eq:estimand}) for $0\le t \le \tau$ in order to obtain (\ref{eq:estimate}) we still have uniform consistency in (\ref{eq:unif}). As a side note the convergence in probability in (\ref{eq:unif}) can be interpreted as ordinary probabilities and not in terms of outer probabilities (cf.\ \cite{vdv2023}) since $\hat{\mu}_E(t)$ and $\mu_E(t)$ are both measurable functions of $t$ on the compact set $t\in[0,\tau]$.

Before presenting the three theorems that comprise the theoretical part of this paper we first make some additional definitions and assumptions. Define $Y(t):=\frac{1}{n}\sum_{i=1}^nY_i(t)$, $N^R(t):=\frac{1}{n}\sum_{i=1}^nN_i^R(t)$, $N^D(t):=\frac{1}{n}\sum_{i=1}^n N_i^D(t)$, $dN^R_P(t):=\frac{1}{n}\sum_{i=1}^nY_i(t)d\Lambda^R_{P}(t;X_i)$, and $N_P^R(t)=\int_0^t dN_P^R(u)$. These four processes are the building blocks of $\hat{\mu}_E=\{\hat{\mu}_E(t); \, 0\le t \le \tau\}$ in (\ref{eq:estimate}). To ensure weak convergence of these four processes and $\hat{\mu}_E$, we make the following assumptions:
\begin{enumerate}
    \item $P(N_i^R(\tau)\leq K_1)=1$ for some $K_1<\infty$.\label{as:nrbound}
    \item $P[\Lambda^R_{P}(\tau;X_i)<K_2]=1$ for some $K_2<\infty$. \label{as:lrbound}
        \item $\Lambda^R_P(t;X_i)$ is \textit{càdlàg} in $t\in[0,\tau]$. \label{as:lrcadlag}
    \item $P(Y_i(\tau)>0)\ge\varepsilon >0$.
    \label{as:ybound} 
    \item $\Lambda^D(t)$ is continuous in $t\in[0,\tau]$ \label{as:ldcont}
\end{enumerate}

 Detailed proofs of the below theorems can be found in the appendix. The first theorem relates to the weak convergence after standardisation of the four processes $N^R(t)$, $N^P(t)$, $N^R_P(t),$ and $Y(t)$ that are the building blocks of $\hat{\mu}_E$. 
\begin{theorem} \label{thrm:comp}
    Assume conditions 1-\ref{as:lrcadlag} above. The right continuous version of the process $\{\sqrt{n}(N^R(t)-E[N^R(t)],N^D(t)-E[N^D(t)],N^R_P(t)-E[N^R_P(t)],Y(t)-E[Y(t)]); \,  0 \leq t \leq \tau\}$ then converges weakly as $n\to\infty$ to a multivariate zero-mean Gaussian process $\{U(t); \, 0\le t\le\tau\}$ on the metric space $(\mathbb{D}^4[0,\tau],||\cdot||_{\infty})$ of four-dimensional càdlàg functions on $[0,\tau]$, equipped with the supremum norm.
\end{theorem} 
The second theorem contains the main result of the paper. We prove this theorem based on theorem \ref{thrm:comp} using the functional delta method (see for instance \cite{vdv2023}). It is possible to prove theorem \ref{thrm:est} using direct calculations, although one would then miss out on the elegance of a functional approach. 
\begin{theorem} \label{thrm:est}
    Assume conditions 1-\ref{as:ybound} above. Then as $n\rightarrow\infty$ the process $\{\sqrt{n}(\hat\mu_E(t)-\mu_E(t));0\leq t \leq\tau\}$ converges weakly to a zero mean Gaussian process $W=\{W(t); \, 0\le t \le \tau\}$, with covariance function $v(s,t)=\mbox{Cov}[W(s),W(t)]$, on $(\mathbb{D}[0,\tau],||.||_\infty)$.
\end{theorem}
Our last theorem follows more or less directly from the proof of theorem \ref{thrm:est} together with some additional arguments regarding consistency. The result is almost the same as that presented in Ghosh \& Lin \parencite*{gosh2000} with some differences related to the relative recurrence estimator. 

\begin{theorem} \label{thrm:cov}
    Assuming conditions 1-\ref{as:ldcont}, the covariance function $\sigma^2(s,t)=\sigma_n^2(s,t)\\=\mbox{Cov}[\hat{\mu}_E(s),\hat{\mu}_E(t)]$ of the estimator $\hat{\mu}_E=\{\hat{\mu}_E(t);\, 0\le t \le \tau\}$ can be approximated by $\sum_{i=1}^n E[\Phi_i(s)\Phi_i(t)]/n^2$, where the terms $\Phi_i(t)$ are defined in the proof. This gives rise to an estimate 
    \begin{align}
        \label{eq:hsi}
        \hat  \sigma^2(s,t)&=\frac{1}{n^2}\sum_{i=1}^n\hat\Phi_i(s)\hat\Phi_i(t)
\end{align}
    of the covariance function, with
    \begin{align}
        \label{eq:hPhii}
        \hat \Phi_i(t)&=\int^t_0 \hat S(u)\frac{dN_i^R(u)-Y_i(u)d(\Lambda^R_{P}(u;X_i)+\hat \Lambda^R_{E}(u))}{Y(u)}\nonumber\\ &- \hat\mu_{E}(t)\int_0^t\frac{dN_i^D(u)-Y_i(u)d\hat\Lambda^D(u))}{Y(u)} \\
        &+\int_0^t\hat\mu_{E}(u)\frac{dN_i^D(u)-Y_i(u)d\hat\Lambda^D(u)}{Y(u)}.\nonumber
\end{align}
Moreover, $\hat{v}(s,t)=n\hat{\sigma}^2(s,t)$ is a consistent estimator of the asymptotic covariance function $v(s,t)$ of the Gaussian process $W$ in theorem \ref{thrm:est} as $n\to\infty$. The convergence is uniform for all $(s,t)$ in the sense that
\begin{align}
\label{eq:hvv}
\max_{0\le s,t \le \tau} |\hat{v}(s,t)-v(s,t)| \conp0
\end{align}
as $n\to\infty$. 
\end{theorem}
Theorem \ref{thrm:cov} allows for the construction of confidence intervals of $\mu_E(t)$, of the form 
$$
I(t)=\hat \mu_E(t) \pm z_{1-\alpha/2} \hat\sigma(t,t),
$$
where $z_{1-\alpha/2}$ is the $(1-\alpha/2)$-percentile from the standard normal distribution for the desired (nominal) significance level $\alpha$. It follows from theorem \ref{thrm:cov} that the asymptotic coverage probability of $I(t)$ as $n\to\infty$ is $1-\alpha$. 

It is worth noting that the estimator $\hat \mu_E(t)$ is not necessarily positive, making confidence intervals based on a log transform inadvisable. This might be seen as a downside of the estimator but we would argue that it is not. The estimator is simply comparing the number of recurrent events individuals have with what we expect based on the general population. Should this number be negative the interpretation is simply that the sampled individuals have fewer events than what we expected them to have i.e.\ some sort of recurrent events deficit. The multiplication with the estimated survival curve $\{\hat{S}(u); 0\le u \le \tau\}$ in the integrand of $\hat{\mu}_E(t)$ does not change the interpretation of the estimated relative recurrences in (\ref{eq:estimate}), since individuals after their death will always have zero expected number of recurrent events. 

\section{Simulation study}\label{sec:sim}
We investigated the finite sample properties of our estimator (\ref{eq:estimate}) by simulations. Anticipating application of the method using registry data where individuals are enrolled into a study at the date of diagnosis, we created simulations that mimic such a scenario. Computations were done in C++ and R through the Rcpp interface \parencite{rcpp2026}.

The data generating process is as follows. All processes were simulated in continuous time but event times were rounded to integers representing days, which reflects processes developing in continuous time but recorded in discrete time. Date of diagnosis was drawn from a uniform distribution taking values between the dates 2020-01-01 and 2023-12-31. End of study was set as 2024-12-31 for all individuals. The censoring variable $C_i$ was thus a deterministic function of the random date of diagnosis, making it not independent of the recurrent events process but non-informative. Age at diagnosis was drawn from a uniform distribution ranging from 30 to 90 and then rounded, and sex was drawn from a Bernoulli distribution with probability 0.5. Together age, date of diagnosis, and sex form the random vector $X_i$ described in section \ref{sec:est}. Based on this vector, hospitalisation rates were extracted from tables detailing the number of hospitalisations per age, sex, and year from the National Board of Health and Welfare (\url{www.socialstyrelsen.se/en}). The hospitalisation rates were continuously updated whenever the calender year changed or the individual aged a year (with birthday taken to be date of diagnosis). A disease severity random effect was drawn from a mean zero normal distribution corresponding to the variable $Z_i$. 

To simulate the two event processes we defined a design vector $W_i$ containing 1, the mean standardised age, and severity of individual $i$. The hazard rate was then given by $d\Lambda^D(t;W_i)=\exp(\beta_D'W_i)dt$ and the \textit{excess} recurrent events rate by $d\Lambda^R_E(t;W_i)=\exp(\beta_R'W_i)dt$, where prime refers to vector transposition and the column vectors $\beta_D$ and $\beta_R$ are given in table \ref{tab:simscen}. The total recurrent events rate $d\Lambda^R(t;W_i)$ was then obtained by adding the excess rate $d\Lambda^R_E(t;W_i)$ with the population rate $d\Lambda^R_P(t;X_i)$ taken from the aggregated data. 

Since the hazard rate is constant, time to death was generated from an exponential distribution with an individual specific rate $\exp(\beta_D^\prime W_i)$. The generation of recurrent events times was more intricate due to the rates being piecewise constant. For each time interval $(t_{k-1},t_k)$ of constant rates for an individual $i$ the following algorithm was used:
\begin{enumerate}
    \item Generate the number of events $N_i(t_{k-1},t_k)$ in the interval from a Poisson distribution.
    \item If $N_i(t_{k-1},t_k)>0$ generate the sum of the event times from a truncated Gamma distribution.
    \item Generate exponential times based on the sum of event times as described in Lindqvist \& Taraldsen \cite{lindqvist2007}.
\end{enumerate}

The scenarios investigated can be found in table \ref{tab:simscen}. There $\beta_D$ and $\beta_R$ denote the coefficient vector of the hazard rate and excess recurrent events rate respectively, and $\sigma$ denotes the standard deviation of the severity random effect. We chose to vary the baseline hazard rate and baseline recurrent events rate since different diseases tend to differ in terms of mortality and in terms of the number of recurrent events. The two values of the standard deviation of the random effect were chosen to investigate the impact of increased variability and increased correlation between events. The scenarios were investigated for the sample sizes $n=500$ and $n=2000$ and each scenario was repeated 5000 times per sample size. This should give a Monte Carlo error of around 0.3 percentage units, for the true coverage probability of confidence intervals with a nominal coverage of 95\%. True values were created by running through the scenarios with a sample size of 100,000. 

\begin{table}[h]
    \caption{Simulation scenarios. The column vectors $\beta_D$ and $\beta_R$ contain the coefficients of the models for the hazard rate and the \textit{excess} recurrent events rate respectively. The elements from left to right correspond to the intercept, the effect of mean standardised age, and the effect of severity. The parameter $\sigma$ is the standard deviation of the mean zero random effect representing severity.} 
    \label{tab:simscen} 
    \centering
    \begin{tabular}{c|ccc}
        \hline
        Scenario & $\beta_D'$&$\beta_R'$ &$\sigma$ \\
        \hline
        1 & (-7,0.01,0.25) & (-6,0.01,0.5) & 1 \\
        2 & (-7,0.01,0.25) & (-7,0.01,0.5) & 1 \\
        3 & (-8,0.01,0.25) & (-6,0.01,0.5) & 1 \\
        4 & (-8,0.01,0.25) & (-7,0.01,0.5) & 1 \\
        5 & (-7,0.01,0.25) & (-6,0.01,0.5) & 2 \\
        6 & (-7,0.01,0.25) & (-7,0.01,0.5) & 2 \\
        7 & (-8,0.01,0.25) & (-6,0.01,0.5) & 2 \\
        8 & (-8,0.01,0.25) & (-7,0.01,0.5) & 2 \\
        \hline
    \end{tabular}

\end{table}

The results can be found in table \ref{tab:simres500} below and table \ref{tab:simres2000} in the appendix. We see that the expected number of excess recurrent events estimator $\hat{\mu}_E(t)$, and its standard error $\hat{\sigma}(t,t)$, perform well in finite samples. The actual coverage of the confidence intervals $I(t)$ is generally close to their nominal coverage for both the smaller and the larger sample size. Some tendencies of overestimation of the standard deviation can be seen at the latest time point in scenarios with higher mortality (scenarios 1, 2, 5, and 6). This is likely due to the low number of individuals at risk.

\begin{table}[hbt]
   \caption{Simulation results with $n=500$. The columns $\mbox{SE}(\hat\mu_E(t))$ and $\mbox{SD}(\hat\mu_E(t))$ represent the mean of the standard errors $\hat{\sigma}(t,t)$ and the mean of the true standard deviations $\sigma(t,t)$ respectively,
   over 5000 replicates. 
   The rightmost column $1-\alpha_{\mbox{\scriptsize true}}$ represents the actual coverage of confidence intervals $I(t)$ with nominal coverage $1-\alpha=95\%$. Time is given in days.}
\centering
\begin{tabular}{ccc|cccc}
  \hline
 Scenario & Time $t$ & $\mu_E(t)$ & $\hat\mu_E(t)$ & $\mbox{SE}(\hat\mu_E(t))$ & $\mbox{SD}(\hat\mu_E(t))$ & $1-\alpha{\mbox{\scriptsize true}}$ \\ 
  \hline
1 & 182 & 0.51 & 0.52 & 0.04 & 0.04 & 0.95 \\ 
   & 365 & 0.85 & 0.85 & 0.06 & 0.05 & 0.95 \\ 
   & 730 & 1.42 & 1.42 & 0.09 & 0.08 & 0.98  \\ \hline
  2 & 182 & 0.19 & 0.19 & 0.03 & 0.02 & 0.95 \\ 
   & 365 & 0.31 & 0.32 & 0.04 & 0.03 & 0.96 \\
   & 730 & 0.53 & 0.53 & 0.06 & 0.05 & 0.98 \\  \hline
  3 & 182 & 0.55 & 0.55 & 0.04 & 0.04 & 0.94 \\ 
   & 365 & 0.96 & 0.96 & 0.06 & 0.06 & 0.94 \\ 
   & 730 & 1.78 & 1.78 & 0.09 & 0.09 & 0.96 \\  \hline
  4 & 182 & 0.20 & 0.20 & 0.03 & 0.03 & 0.94 \\ 
   & 365 & 0.35 & 0.36 & 0.04 & 0.04 & 0.95 \\  
   & 730 & 0.65 & 0.66 & 0.06 & 0.05 & 0.96 \\ \hline
  5 & 182 & 0.70 & 0.71 & 0.06 & 0.06 & 0.95 \\ 
   & 365 & 1.11 & 1.12 & 0.09 & 0.08 & 0.95 \\ 
   & 730 & 1.73 & 1.76 & 0.14 & 0.12 & 0.97 \\ \hline
  6 & 182 & 0.26 & 0.26 & 0.03 & 0.03 & 0.96 \\ 
   & 365 & 0.41 & 0.42 & 0.04 & 0.04 & 0.96 \\ 
   & 730 & 0.65 & 0.65 & 0.07 & 0.06 & 0.97 \\ \hline
  7 & 182 & 0.78 & 0.79 & 0.06 & 0.06 & 0.94 \\ 
   & 365 & 1.33 & 1.34 & 0.10 & 0.10 & 0.94 \\ 
   & 730 & 2.37 & 2.39 & 0.17 & 0.16 & 0.95 \\  \hline
  8 & 182 & 0.29 & 0.29 & 0.03 & 0.03 & 0.95 \\ 
   & 365 & 0.49 & 0.50 & 0.05 & 0.05 & 0.94 \\ 
   & 730 & 0.89 & 0.89 & 0.08 & 0.08 & 0.95 \\  \hline
\end{tabular}
 \label{tab:simres500}
\end{table}

\section{Application} \label{sec:app}
We illustrate the usage of the excess recurrent events estimator using data from CRCBaSe, a Swedish registry containing data on patients with colon or rectal cancer. This database contains data on around 130,000 patients diagnosed with colon or rectal cancer taken from SCRCR, the national quality registry for colorectal cancer. This information is supplemented with data from other registries, including the national in-patient registry wherein data on hospital admissions can be found. For more information on CRCBaSe and its data sources see Weibull et al. \parencite*{Weibull2023}. 

From CRCBaSe we extracted a cohort of rectal patients diagnosed between 2015-01-01 and 2024-12-31. The cohort was restricted to those with no prior colorectal cancer diagnosis. We included patients between 18 and 69 years old at diagnosis with pathologic cancer stage 3 or 4, with the a priori belief that these patients would be burdened by many hospitalisations. 

Time to death or censoring was defined as the number of days from diagnosis to the date of death or end of study (2024-12-31). This implicitly assumed no loss to follow-up, although since this would only occur due to migration, the impact of this assumption is likely to be minor. Time of hospitalisation was defined similarly, with censoring at time of death or end of study. To avoid counting movements between hospitals or hospital wards as new events, we only counted admission records that happened at least 1 days from the last discharge. 

Population hospitalisation rates were again taken from the aggregated data compiled by the National Board of Health and Welfare (\url{www.socialstyrelsen.se/en}), which is available in 5-year age categories on the website. 

In total 1749 patients fulfilled the inclusion criteria. Out of them $35.5\%$ were female. The median age at diagnosis was 62 and the 1st and 3rd quartiles were 55 and 66 indicating a fairly low spread of ages. The percentage of patients with stage 3 was $24.5\%$. 

The estimated number $\hat{\mu}_E(t)$ of excess hospitalisations is displayed in figure \ref{fig:exchosp} as a function of $t$ (given in days). At the last day of follow-up (no individuals were at risk beyond 3600 days) the estimated number of excess hospitalisations was $4.87$, with a 95\% confidence interval $(4.02, 5.73)$. This indicates a substantial burden on both patients and the healthcare system. In figure \ref{fig:exchosp} we see an  accumulation of excess events that decreases in speed as time passes. The 75th percentile survival time was 1228 days as estimated by the product limit estimator. This indicates that the accumulation of excess events does not decrease because of decreasing hospitalisation rates, but because of the high mortality.

While the number of excess hospitalisations is high, we must keep in mind that some hospitalisations could be records related to cancer treatment (e.g.\ surgery) and thus not a purely negative event. Treatment related admissions are especially likely during the first year since diagnosis. The results suggest that there is a substantial burden in terms of healthcare use in this patient group and that earlier detection could be of considerable value both to society and patients.

\begin{figure}[h]
    \centering
    \includegraphics[alt={Graph of the number of cumulative number of excess hospitalisations from diagnosis up to 5 years. The number increases with a diminishing rate.},width=1\linewidth]{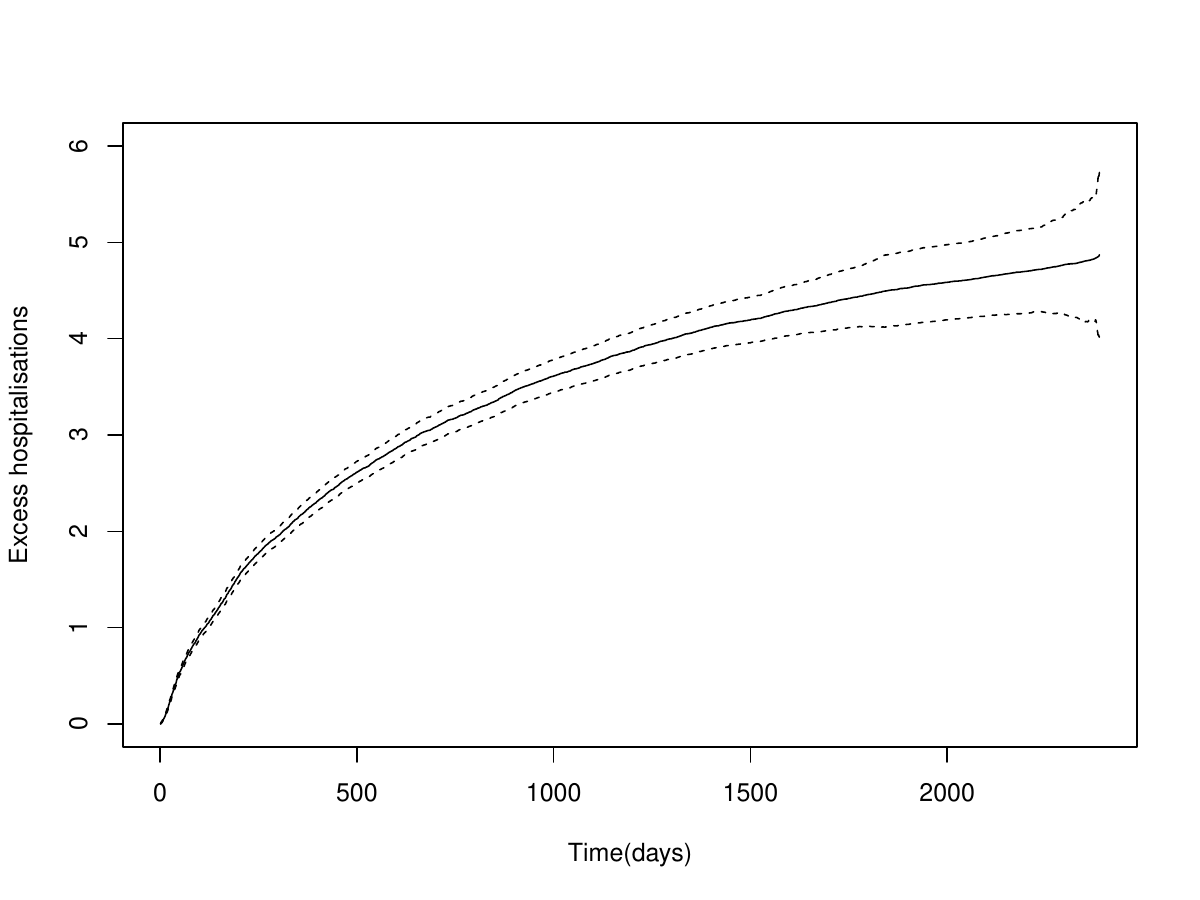}
    \caption{The estimated number $\hat{\mu}_E(t)$ of excess hospitalisations in patients with stage 3-4 rectal cancer, as a function of $t$ (solid line) and the accompanying confidence intervals $I(t)$ with a nominal coverage of $95\%$ (dashed lines).}
    \label{fig:exchosp}
\end{figure}

\section{Discussion}\label{sec:disc}
We have extended the concept of relative survival to the field of recurrent events by presenting a non-parametric estimator for the excess number of recurrent events. This estimator allows for the estimation of the societal burden e.g.\ in terms of hospitalisations that a particular disease causes. We have furthermore proven weak convergence of the estimator as a functional of stochastic processes and investigated its finite sample properties through simulations. Lastly, we illustrated how the estimator might be used in Swedish data on patients with rectal cancer.

It is our belief that this estimator will prove useful in quantifying the burden of illness of different diseases. In particular, the estimator is highly relevant for those diseases that are characterised both by a large number of some recurrent events and high mortality. Examples of such diseases include certain cancer types, as illustrated here, heart failure, and chronic pulmonary disease.

There are certain aspects regarding the interpretation of our estimator that are worth some discussion. Our derivations are based on a marginal interpretation of the estimator, i.e.\ averaging over both disease related and non-disease related patient characteristics. This provides a reasonable estimate of the total burden of disease which is extremely valuable to researchers in public health and health economics. It is perhaps of lesser value to individual patients and for decision making by physicians, since we would then have at least partial information on the characteristics of the particular patient. This can be partially remedied through stratification, e.g.\ on age and sex, provided that the sample size in each stratum is large enough. Regression modelling might however be preferable.

Another aspect related to the interpretation our estimator is that it does not produce a causal estimate in the sense of what would happen should one be able to isolate the impact of the disease. The reason being that overall survival is part of the estimand. This is a conscious choice on our behalf, since we believe that there is greater benefit in being able to quantifying something that is fully interpretable in the real world - the excess number of events in comparison to the general population taking death into account. It is nonetheless certainly possible to replace the product limit estimator (\ref{eq:kmcomp}) in (\ref{eq:estimate}) with a suitable estimator of relative survival. This could under rather specific assumptions allow for the isolation of both excess recurrent events and excess mortality. For example, one would have to assume that the population rates are unaffected by disease related events. Weak convergence could be proven in a similar way to what is presented in this article.  

Finally, our estimator requires access to population rates of sufficient granularity. As stated earlier, population rates are publicly available in Sweden in the case of hospitalisations. If that is not the case in the country of interest, it is possible that such data can be requested from relevant authorities. Another option is to estimate population rates from some other data or to estimate them from matched controls \parencite{Bower2017}. Additionally, the interpretability of estimates is affected by how well the individuals' characteristics are reflected in the individuals the population rates are based on. Often population rates are stratified on age, sex, and year, which may be insufficient for certain diseases. In such situations, the usage of matched controls may be necessary.

\section{Declaration of the use of generative AI and AI-assisted technologies}
During the preparation of this work the author(s) used Microsoft 365 Copilot in order to get suggestions on how to sort out different problems encountered during programming. After using this tool/service the author(s) reviewed and edited the content as necessary and take(s) full responsibility for the content of the publication. 

\printbibliography

\appendix
\section{Appendix} 
\subsection{Proof of theorem 1}
\begin{proof} 
     It is convenient to introduce the process $U_n = \{U_n(t); \, 0\le t\le\tau\}$, with
     \begin{align}
     \label{eq:Un}
     U_n(t) &= \left[\sqrt{n}(N^R(t)-E(N^R(t)),\sqrt{n}(N^D(t)-E(N^D(t))), \right.\nonumber\\
     &= \left. \sqrt{n}(N_P^R(t)-E(N_P^R(t)),\sqrt{n}[Y(t)-E(Y(t))]\right]\\
     &=: (U_{n1}(t),U_{n2}(t),U_{n3}(t),U_{n4}(t)).\nonumber
     \end{align}
     Note that only the first three components of $U_n$ are \textit{càdlàg} (right continuous with left-hand limits), whereas the last component $U_{n4}=\{U_{n4}(t);\, 0\le t \le\tau\}$ is left continuous with right-hand limits, with a finite number of discontinuities. Without loss of generality we assume that $U_{n4}$ is modified at these discontinuity points, so that $U_{n4}$ and $U_n$ are random elements of $(\mathbb{D}[0,\tau],||\cdot||_{\infty})$ and $(\mathbb{D}^4[0,\tau],||\cdot||_{\infty})$ respectively.
     Since the processes $N^R(t)=\sum_{i=1}^n N_i^R(t)/n$, $N^D(t)=\sum_{i=1}^n N_i^D(t)/n$, $N_P^R(t)=\sum_{i=1}^n \int_0^t Y_i(u)d\Lambda_P^R(u;X_i)/n$ and $Y(t)=\sum_{i=1}^n Y_i(t)/n$ are averages of iid terms, convergence of finite-dimen\-sional distributions of each $U_{nj}=\{U_{nj}(t);\, 0\le t \le \tau\}$ to multivariate normal distributions is a consequence of the Central Limit Theorem.   
     It follows from conditions 1-\ref{as:lrcadlag} that the terms $N_i^R(t)$, $N_i^D(t)$, and $\int_0^t Y_i(u)d\Lambda_P^R(u;X_i)$ of the first three processes are monotone increasing and \textit{càdlàg}  in $t$ and uniformly bounded on $[0,\tau]$ with probability 1. Likewise, the version of $-Y_i(t)$ that is modified at discontinuity points is also monotone increasing and \textit{càdlàg}. 
     Invoking Example 2.11.16 of van der Vaart \& Wellner \parencite*{vdv2023} we conclude that $\{U_{nj}\}_{n\ge 1}$ are asymptotically tight processes on $(\mathbb{D}[0,\tau],||\cdot||_{\infty})$ for $j=1,2,3,4$. Together with the above mentioned convergence of finite-dimensional distributions, this proves that each $U_{nj}$ converges weakly on $(\mathbb{D}[0,\tau],||\cdot||_{\infty})$ as $n\to\infty$ to a zero mean Gaussian processes $U_j=\{U_j(t);\, 0\le t \le \tau\}$ for $j=1,2,3,4$. Joint weak convergence of all four processes $U_n=(U_{n1},U_{n2},U_{n3},U_{n4})$ on $(\mathbb{D}^4[0,\tau],||\cdot||_{\infty})$, as $n\to\infty$, to a four-dimensional Gaussian process $U=\{U(t)=(U_1(t),U_2(t),U_3(t),U_4(t); \, 0\le t \le\tau\}$
     follows from an application of the Cramér-Wold device.
\end{proof}
\subsection{Proof of theorem 2}
\begin{proof}
    Assuming a multiplicative intensity model we have the following expectations of the four processes involved in Theorem \ref{thrm:comp}, due to our assumption of an iid sample, and $Y_i(t)$ only taking values in $\{0,1\}$:
    \begin{align*}
     E[dN^R(t)] &= \frac{1}{n}\sum_{i=1}^n E[E[dN_i^R(t)|Y_i(t),X_i,Z_i]]\\
     &= \frac{1}{n}\sum_{i=1}^n E[Y_i(t)\{d\Lambda_E^R(t;X_i,Z_i) + d\Lambda_P^R(t;X_i)\}]\\
     &= E[Y_i(t)\{d\Lambda_E^R(t;X_i,Z_i) + d\Lambda_P^R(t;X_i)\}]\\
     &= [d\Lambda_E^R(t)+d\Lambda_P^R(t)]C(t)S(t),\\
    E[dN^R_P(t)]&= \frac{1}{n}\sum_{i=1}^n E[Y_i(t)d\Lambda^R_{P}(t;X_i)] \\
    &= E[Y_i(t)d\Lambda^R_{P}(t;X_i)] \\
    &= d\Lambda_P^R(t)C(t)S(t),\\
    E[dN^D(t)]&= \frac{1}{n}\sum_{i=1}^n E[Y_i(t)d\Lambda^D(t;X_i,Z_i)],\\
    &= E[Y_i(t)]d\Lambda^D(t;X_i,Z_i)], \\
    &= d\Lambda^D(t)C(t)S(t),\\
    E[Y(t)] &= \frac{1}{n}\sum_{i=1}^n E[Y_i(t)]\\
    &= E[Y_i(t)]\\
    &= C(t)S(t).
   \end{align*}
  Letting $N^R$, $N^R_P$, $N^D$ and $Y$ be represented by $A$, $B$, $G$ and $D$ the estimator $\hat{\mu}_E=\{\hat{\mu}_E(t);\, 0\le t \le\tau\}$ consists of the following maps:
    \begin{align*}
        &(A,B,G,D) \\
        &\mapsto \left(A-B,G,\frac{1}{D}\right) \\
        &\mapsto \left(\int\frac{d(A-B)}{D},\int\frac{dG}{D} \right)\\
        &\mapsto \left(\int\frac{d(A-B)}{D}, \prodi \left(1-\frac{dG}{D}\right) \right)\\
        &\mapsto \int\frac{d(A-B)}{D}\prodi \left(1-\frac{dG}{D}\right)
    \end{align*}
    Replace $(A,B,G,D)$ with $(A_0+\alpha,B_0+\beta,G_0+\gamma,D_0+\delta)$ where the Greek letters represent departures of $A$, $B$, $G$ and $D$ from $A_0=E(A)$, $B_0=E(B)$, $G_0=E(G)$, and $D_0=E(D)$ respectively. The above maps are all compactly differentiable at $\alpha=\beta=\gamma=\delta=0$ with derivative maps:
    \begin{align*}
        &(\alpha,\beta,\gamma,\delta) \\
        &\mapsto \left(\alpha-\beta,\gamma,-\frac{\delta}{D_0^2}\right) \\
        &\mapsto \left(\int\frac{d(\alpha-\beta)}{D_0} -\int\frac{\delta}{D_0^2}d(A_0-B_0), \int\frac{d\gamma}{D_0} -\int\frac{\delta}{D_0^2}dG_0\right)\\
        &\mapsto \left(\int\frac{d(\alpha-\beta)}{D_0} -\int\frac{\delta}{D_0^2}d(A_0-B_0), \right.
        \\
        &\,\,\,\,\,\,\,\,\,\, \left. - \prodi \left(1-\frac{dG_0}{D_0}\right)\int\frac{1}{1-\frac{\Delta G_0}{D}} \left[\frac{d\gamma}{D_0} -\frac{\delta}{D_0^2}dG_0\right]\right) \\
        &\mapsto \int \prodi \left(1-\frac{dG_0}{D_0}\right)\left[\frac{d(\alpha-\beta)}{D_0} -\frac{\delta}{D_0^2}d(A_0-B_0)\right] \\
        &\,\,\,\,\,\,\,\,\,\, - \int \prodi\left(1-\frac{dG_0}{D_0}\right) \int\frac{1}{1-\frac{\Delta G_0}{D_0}} \left[\frac{d\gamma}{D_0} -\frac{\delta}{D_0^2}d G_0\right]\frac{d(A_0-B_0)}{D_0}
    \end{align*}
    where $\Delta X=\lim_{\varepsilon \rightarrow 0} X(t)-X(t-\epsilon)$. Naturally, if $X$ is (left) continuous then $\Delta X=0$. For a proof of compact differentiability of the product integral see the excellent exposition by Gill \& Johansen \parencite*{gill1990}.
    
   Note that the right-hand-side of the above displayed equation is a linearized approximation of $\hat{\mu}_E(t)-\mu(t)$ when $t$ is the upper limit of integration of all outer integrals. Multiplying this asymptotically accurate approximation of $\hat{\mu}_E(t)-\mu(t)$ with $\sqrt{n}$, and making use of $D_0(t)=E[Y(t)]=C(t)S(t)$, $dA_0(t)-dB_0(t)=d\Lambda_E^R(t)C(t)S(t)$, $dG_0(t)=d\Lambda^D(t)C(t)S(t)$, and $dG_0(t)/D_0(t)=d\Lambda^D(t)$, we have:
        \begin{align*}
            &W_n(t) = \sqrt{n}[\hat{\mu}_E(t)-\mu(t)] \\
            &\approx \sqrt{n}\int_0^t\prodi_{[0,u]} \left(1-d\Lambda^D(s)\right) \left[ \frac{dN^R(u)-dN_P^R(u) - C(u)S(u)d\Lambda^R_E(u)}{C(u)S(u)} \right. \\
            &- \left. \frac{Y(u)-C(u)S(u)}{[C(u)S(u)]^2}C(u)S(u)d\Lambda_E^R(u)\right] \\
            &-\sqrt{n}\int_0^t \prodi_{[0,u]} \left(1-d\Lambda^D(s)\right) \int_0^u\frac{1}{1-\Delta\Lambda^D(s)} \left[\frac{dN^D(s)- C(s)S(s)d\Lambda^D(s)}{C(s)S(s)}  \right. \\  &\left.-\frac{Y(s)-C(s)S(s)}{[C(s)S(s)]^2}C(s)S(s)d\Lambda^D(s)\right]\frac{C(u)S(u)d\Lambda_E^R(u)}{C(u)S(u)} \\
            &= \sqrt{n}\int_0^t\prodi_{[0,u]} \left(1-d\Lambda^D(s)\right) \left[ \frac{dN^R(u)-dN^R_P(u)-Y(u)d\Lambda^R_E(u)}{C(u)S(u)}\right. \\ &\left.- \int_0^u \frac{1}{1-\Delta\Lambda^D(s)} \frac{dN^D(s) - Y(s)d\Lambda^D(s)}{C(s)S(s)}d\Lambda_E^R(u)\right]\\
            &=\sqrt{n}\int_0^tS(u)\left[\frac{dM^R(u)}{C(u)S(u)} - \int_0^u \frac{1}{1-\Delta\Lambda^D(s)} \frac{dM^D(s)}{C(s)S(s)}d\Lambda_E^R(u)\right] \\
            &=:\sqrt{n}\Phi(t),
        \end{align*}
        where in the last step we introduced 
        \begin{align}
        \label{eq:dMR}
        dM^R(t) &= dN^R(t)-\frac{1}{n}\sum_{i=1}^n Y_i(t)d\Lambda^R_P(t;X_i)-Y(t)d\Lambda^R_E(t)\nonumber\\
        &= dN^R(t) - dN_P^R(t) - Y(t)d\Lambda^R_E(t)\nonumber\\
        &= d[N^R(t)-E(dN^R(t))] - [dN_P^R(t)-E(dN_P^R(t))]\\
        & - [Y(t)-E(Y(t))]d\Lambda^R_E(t)\nonumber\\
        & = [dU_{n1}(t) - dU_{n3}(t) - U_{n4}(t)d\Lambda^R_E(t)]/\sqrt{n}
        \nonumber        
        \end{align}
        and 
        \begin{align}
        \label{eq:dMD}
        dM^D(t) &= dN^D(t)-\frac{1}{n}\sum_{i=1}^n Y_i(t)d\Lambda^D(t)\nonumber \\
        &= dN^D(t) - Y(t)d\Lambda^D(t)\\
        &= d[N^D(t)-E(N^D(t))] - [Y(t)-E(Y(t))]d\Lambda^D(t)\nonumber\\
        &= [dU_{n2}(t) - U_{n4}(t)d\Lambda^D(t)]/\sqrt{n}.\nonumber
        \end{align}
        Note from (\ref{eq:dMR})-(\ref{eq:dMD}) that $V_n=\{\sqrt{n}(M^R(t),M^D(t)); \, 0\le t \le\tau\}$ is a linear functional of the four-dimensional process $U_n=\{(U_{n1}(t),\ldots,U_{n4}(t);\, 0\le t\le\tau\}$ of (\ref{eq:Un}), that was introduced in the proof of theorem \ref{thrm:comp}.
        It follows from weak convergence of $U_n$ as $n\to\infty$, towards the zero mean Gaussian process $U=\{U(t);\, 0\le t\le\tau\}$ in theorem \ref{thrm:comp}, condition 3 above theorem \ref{thrm:comp}, and the Continuous Mapping Theorem that $V_n$ converges weakly to a two-dimensional zero mean Gaussian  process $V=\{V(t)=(V_1(t),V_2(t)); \, 0\le t \le \tau\}$ on $(\mathbb{D}^2[0,\tau],||\cdot||_{\infty})$ as $n\to\infty$, with
        \begin{align*}
        dV_1(t) &= dU_1(t)-dU_3(t)-U_4(t)d\Lambda_E^R(t),\\
        dV_2(t) &= dU_2(t)-U_4(t)d\Lambda^D(t).
        \end{align*}
        Making use of the functional delta method \parencite[see e.g.][ ]{vdv1998}, that the approximation $\sqrt{n}\Phi(t)$ of $W_n(t)$ is a linear functional of the process $V_n$, and once again the Continuous Mapping Theorem, we conclude that $W_n=\{W_n(t); t \in[0,\tau]\}$ converges weakly to the zero mean Gaussian process $W=\{W(t); \, 0\le t\le\tau\}$ on $(\mathbb{D}[0,\tau],||\cdot||_{\infty})$ as $n\to\infty$, defined as
        $$
        W(t) = \int_0^tS(u)\left[\frac{dV_1(u)}{C(u)S(u)} - \int_0^u \frac{1}{1-\Delta\Lambda^D(s)} \frac{dV_2(s)}{C(s)S(s)}d\Lambda_E^R(u)\right], 
        $$
        and with covariance function $v(s,t)=\mbox{Cov}[W(s),W(t)]=E[W(s)W(t)]$.
\end{proof}            

\subsection{Proof of theorem 3}
\begin{proof} 
We will start by deriving an expression for the covariance function $\sigma^2(s,t)$ of the estimator $\hat{\mu}_E$ before we present the estimator (\ref{eq:hsi})-(\ref{eq:hPhii}) of this covariance function. Because of condition 5, the cumulative hazard function $t\to \Lambda^D(t)$ is assumed to be continuous. For this reason, the expression for $\sqrt{n}\Phi(t)$ in the proof of theorem \ref{thrm:est} simplifies to 
    \begin{align}
        \label{eq:Phit}
        \sqrt{n}\Phi(t)=\sqrt{n}\int_0^tS(u)\left[\frac{dM^R(u)}{C(u)S(u)} - \int_0^u \frac{dM^D(s)}{C(s)S(s)}d\Lambda_E^R(u)\right].
    \end{align}
    Through (\ref{eq:estimand}) and integration by parts we can rewrite (\ref{eq:Phit}) as 
    \begin{align}
        \label{eq:Phit2}
        \sqrt{n}\Phi(t)&=\sqrt{n}\left[\int^t_0\frac{dM^R(u)}{C(u)} - \mu_E(t)\int_0^t\frac{dM^D(u)}{C(u)S(u)}\right.  \nonumber\\ &\left.+\int_0^t\mu_E(u)\frac{dM^D(u)}{C(u)S(u)} \right].
    \end{align}

Recall from (\ref{eq:dMR})-(\ref{eq:dMD}) that the two processes $dM^R(t)=\frac{1}{n}\sum_{i=1}^n dM_i^R(t)$ and $dM^D(t)=\frac{1}{n}\sum_{i=1}^n dM_i^D(t)$ are averages of iid terms 
\begin{align}
\label{eq:dMiR}
dM_i^R(t)=dN_i^R(t)-Y_i(t)d[\Lambda_P^R(t;X_i)+d\Lambda_E^R(t)]    
\end{align}
and 
\begin{align}
\label{eq:dMiD}
dM_i^D(t)=dN_i^D(t)-Y_i(t)d\Lambda^D(t)
\end{align}
respectively. From this and (\ref{eq:Phit2}) it follows that $\Phi(t)$ is an average  
    \begin{align}
        \label{eq:Phi}
        \Phi(t)=\frac{1}{n}\sum_{i=1}^n\Phi_i(t)
    \end{align}
    of iid terms as well, with
    \begin{align}
        \label{eq:Phii}
        \Phi_i(t)=\int^t_0 \frac{dM_i^R(u)}{C(u)} - \mu_E(t)\int_0^t\frac{dM_i^D(u)}{C(u)S(u)} +\int_0^t\mu_E(u)\frac{dM_i^D(u)}{C(u)S(u)}.
    \end{align}
Based on these results we conclude that the covariance function of the estimation error process $\{\hat{\mu}_E(t)-\mu_E(t); \, 0\le t \le \tau\}$ can be approximated with
\begin{align}
    \label{eq:sigexp}
    \sigma^2(s,t) 
    &= \mbox{Cov}[\hat{\mu}_E(s)-\mu_E(s),\hat{\mu}_E(t)-\mu_E(t)]\nonumber\\
    &\approx E[\Phi(s)\Phi(t)]\\
    &= \sum_{i=1}^nE[\Phi_i(s)\Phi_i(t)]/n^2\nonumber
\end{align}
for $s,t\in[0,\tau]$. Equations (\ref{eq:dMiR})-(\ref{eq:dMiD}) and (\ref{eq:Phii})-(\ref{eq:sigexp}) suggest (in agreement with (\ref{eq:hsi})-(\ref{eq:hPhii})) a covariance estimator
\begin{align} \label{eq:covest}
    \hat  \sigma^2(s,t) =\frac{1}{n^2}\sum_{i=1}^n\hat\Phi_i(s)\hat\Phi_i(t),
\end{align}
where 
\begin{align}
    \label{eq:hPhii2}
    \hat \Phi_i(t)&=\int^t_0 \hat S(u)\frac{dN_i^R(u)-Y_i(u)d[\Lambda^R_{P}(u;X_i)+\hat \Lambda^R_E(u)]}{Y(u)}\nonumber\\ &- \hat\mu_E(t)\int_0^t\frac{dN_i^D(u)-dY_i(u)\hat \Lambda^D(u)}{Y(u)} \\
    &+\int_0^t\hat\mu_E(u)\frac{dN_i^D(u)-Y_i(u)d\hat \Lambda^D(u)}{Y(u)}\nonumber
\end{align}
is an estimate of $\Phi_i(t)$. The corresponding estimator of the covariance function $v(s,t)=\mbox{Cov}[W(s),W(t)]$, of the limiting Gaussian process $W=\{W(t);\, 0\le t \le\tau\}$ of theorem \ref{thrm:est}, is
\begin{align}
    \label{eq:hv}
    \hat{v}(s,t) = n\hat{\sigma}^2(s,t) 
    = \frac{1}{n}\sum_{i=1}^n\hat\Phi_i(s)\hat\Phi_i(t).  
\end{align}
In order to prove that $\hat{v}(s,t)$ is a uniformly consistent estimator of $v(s,t)$ as $n\to\infty$ we first deduce from (\ref{eq:Phi}), the fact that $\{\Phi_i(t)\}_{i=1}^n$ are iid with $E[\Phi_i(t)]=E[\Phi(t)]=0$, and the Law of Large Numbers, that
\begin{align}
\label{eq:v}
v(s,t) &= E[W(s)W(t)]\\
&= \lim_{n\to\infty} E[W_n(s)W_n(t)]\\
&= \lim_{n\to\infty} E[\sqrt{n}\Phi(s) \cdot \sqrt{n}\Phi(t)]\nonumber\\
&= \lim_{n\to\infty} \sum_{i=1}^n 
E[\Phi_i(s)\Phi_i(t)]/n\\
&= \mbox{plim}_{n\to\infty} v_n(s,t)
,\nonumber
\end{align}
for each $s,t\in [0,\tau]$, with 
\begin{align}
\label{eq:vn}
v_n(s,t) = \frac{1}{n} \sum_{i=1}^n\Phi_i(s)\Phi_i(t)
\end{align}
and with $\mbox{plim}$ referring to convergence in probability. Interpreting $v_n(s,t)-v(s,t)$ as a random element of $(\mathbb{D}[0,\tau]^2,||\cdot||_{\infty})$, an invariance principle can be used to show (similarly as in the proof theorem \ref{thrm:comp}) that 
\begin{align}
\label{eq:vnv}
\max_{0\le s,t \le \tau} |v_n(s,t)-v(s,t)| \conp0 \end{align}
as $n\to\infty$. From (\ref{eq:dMiR})-(\ref{eq:dMiD}), (\ref{eq:Phii}), (\ref{eq:hPhii2}), theorem \ref{thrm:comp} and the Continuous Mapping Theorem we have that 
\begin{align}
\label{eq:UnifPhiConv}
\max_{i\geq1}\sup_{t\in[0,\tau]}|\hat \Phi_i(t)-\Phi_i(t)|\conp0
\end{align} 
as $n\to\infty$, which in turn implies
\begin{align}
\label{eq:hvvn}
\max_{0\le s,t \le \tau} |\hat{v}(s,t)-v_n(s,t)| \conp0 \end{align}
as $n\to\infty$. Combining (\ref{eq:vnv}) and (\ref{eq:hvvn}) we finally deduce (\ref{eq:hvv}). 
\end{proof}

\subsection{Supplementary tables}
\begin{table}[h]
   \caption{Simulation results with $n=2000$. The columns $\mbox{SE}(\hat\mu_E(t))$ and $\mbox{SD}(\hat\mu_E(t))$ represent the mean of the standard errors $\hat{\sigma}(t,t)$ and the mean of the true standard deviations $\sigma(t,t)$ respectively, over 5000 replicates. The rightmost column $1-\alpha_{\mbox{\scriptsize true}}$ represents the actual coverage of the confidence intervals $I(t)$ with nominal cove\-rage $1-\alpha=95\%$. Time is given in days.}
\centering
\begin{tabular}{ccc|cccc}
  \hline
 Scenario & Time $t$ & $\mu_E(t)$ & $\hat\mu_E(t)$ & $\mbox{SE}(\hat\mu_E(t))$ & $\mbox{SD}(\hat\mu_E(t))$ & $1-\alpha_{\mbox{\scriptsize true}}$\\ 
  \hline
1 & 182 & 0.51 & 0.52 & 0.02 & 0.02 & 0.96 \\ 
   & 365 & 0.85 & 0.85 & 0.03 & 0.03 & 0.96 \\ 
   & 730 & 1.42 & 1.42 & 0.05 & 0.04 & 0.98 \\ \hline
  2 & 182 & 0.19 & 0.19 & 0.01 & 0.01 & 0.95 \\ 
   & 365 & 0.31 & 0.32 & 0.02 & 0.02 & 0.96 \\ 
   & 730 & 0.53 & 0.53 & 0.03 & 0.02 & 0.98 \\ \hline
  3 & 182 & 0.55 & 0.55 & 0.02 & 0.02 & 0.95 \\ 
   & 365 & 0.96 & 0.96 & 0.03 & 0.03 & 0.95 \\ 
   & 730 & 1.78 & 1.78 & 0.05 & 0.05 & 0.96 \\ \hline
  4 & 182 & 0.20 & 0.20 & 0.01 & 0.01 & 0.95 \\ 
   & 365 & 0.35 & 0.36 & 0.02 & 0.02 & 0.95 \\ 
   & 730 & 0.65 & 0.66 & 0.03 & 0.03 & 0.96 \\ \hline
  5 & 182 & 0.70 & 0.71 & 0.03 & 0.03 & 0.95 \\ 
   & 365 & 1.11 & 1.12 & 0.04 & 0.04 & 0.95 \\ 
   & 730 & 1.73 & 1.76 & 0.07 & 0.06 & 0.96 \\ \hline
  6 & 182 & 0.26 & 0.26 & 0.02 & 0.02 & 0.95 \\ 
   & 365 & 0.41 & 0.42 & 0.02 & 0.02 & 0.96 \\ 
   & 730 & 0.65 & 0.65 & 0.03 & 0.03 & 0.97 \\ \hline
  7 & 182 & 0.78 & 0.79 & 0.03 & 0.03 & 0.95 \\ 
   & 365 & 1.33 & 1.34 & 0.05 & 0.05 & 0.94 \\ 
   & 730 & 2.37 & 2.39 & 0.09 & 0.08 & 0.95 \\ \hline
  8 & 182 & 0.29 & 0.29 & 0.02 & 0.02 & 0.95 \\ 
   & 365 & 0.49 & 0.49 & 0.02 & 0.02 & 0.94 \\ 
   & 730 & 0.89 & 0.89 & 0.04 & 0.04 & 0.96 \\ 
   \hline
\end{tabular}
    \label{tab:simres2000}
\end{table}

\end{document}